\documentclass{article}
\usepackage[utf8]{inputenc}
\usepackage{xspace}
\usepackage{todonotes}
\usepackage{mathtools}
\usepackage{algorithm}
\usepackage{algpseudocode}
\usepackage{url}
\usepackage{subcaption}
\usepackage{tikz}
\usetikzlibrary{arrows, fit}
\usetikzlibrary{backgrounds,automata,calc}
\usetikzlibrary{decorations.pathreplacing}
\usepackage{csquotes}
\usepackage{amsmath,amsfonts,amsthm, amssymb}
\usepackage{pgfplots}

\DeclareMathOperator*{\argmin}{arg\,min}
\usepackage{enumerate}
\usepackage{paralist}
\usepackage[shortlabels]{enumitem}

\usepackage[margin = 1.25in]{geometry}
\usepackage{appendix}
\usepackage[capitalize]{cleveref}
\usepackage{thmtools}
\usepackage{mathtools}
\usepackage{natbib}
\usepackage{changepage}

\newcommand{\NSW}{\text{NSW}}
\newcommand{\USW}{\text{USW}}
\newcommand{\MAXUSW}{\text{MAX-USW}}
\newcommand{\MNW}{\text{MNW}}
\newcommand{\MMS}{\text{MMS}}
\newcommand{\R}{\mathbb{R}}
\newcommand{\E}{\mathbb{E}}

\renewcommand{\cal}[1]{\mathcal{#1}}

\newcommand{\false}{\texttt{false}\xspace}
\newcommand{\lorenz}{\texttt{lorenz}}
\newcommand{\LOOP}{\texttt{Loop}}
\newtheorem{theorem}{Theorem}[section]
\newtheorem{lemma}[theorem]{Lemma}

\theoremstyle{definition}

\theoremstyle{definition}

\newtheorem{obs}[theorem]{Observation}
\newtheorem{example}[theorem]{Example}

\title{Yankee Swap: a Fast and Simple Fair Allocation
Mechanism for Matroid Rank Valuations}
\author{Vignesh Viswanathan and Yair Zick \\
University of Massachusetts, Amherst \\
{\texttt{\{vviswanathan, yzick\}@umass.edu}}}
\date{}

\begin{document}

\maketitle
\begin{quote}
    ``It Sounds Mean'' 
    
    \raggedleft --- Angela Martin, ``The Office'' S2E10
\end{quote}
\begin{abstract}
We study the problem of fair allocation of indivisible goods when agents have matroid rank valuations. Our main contribution is a simple algorithm based on the colloquial Yankee Swap procedure that computes provably fair and efficient Lorenz dominating allocations. 
While there exist polynomial time algorithms to compute such allocations, our proposed method improves upon them in two ways.
\begin{inparaenum}[(a)]
    \item Our approach is easy to understand and does not use complex matroid optimization algorithms as subroutines.
    \item Our approach is scalable; it is provably faster than all known algorithms to compute Lorenz dominating allocations.
\end{inparaenum}
These two properties are key to the adoption of algorithms in any real fair allocation setting; our proposed method brings us one step closer to this goal.
\end{abstract}

\section{Introduction}\label{sec:intro}
Fair allocation of indivisible goods is an extremely popular problem in the EconCS community.
We would like to assign a set of indivisible \emph{goods} to a set of \emph{agents} who express \emph{preferences} over the sets of goods (or \emph{bundles}) they receive. 
For example, consider the problem of assigning course slots to students \citep{Budish2011EF1,budish2017coursematch}. 
Each student $i$ has a preference over the set of classes they are assigned. 
Student preferences have implicit constraints on their structure: for example, a student may not take two classes if their schedules overlap (they will simply drop one if assigned conflicting classes). Similarly, students gain no value from receiving two seats in the same class. In addition, universities often impose bounds on the number of classes that a student may take in a given semester, or on the number of classes a student takes from a given major.  
We wish to identify an assignment of course seats to students that is both \emph{efficient} --- no student wants any additional available slots --- and \emph{fair} --- no student would prefer another student's assignment to their own. 
Course allocation can thus be naturally cast as an instance of a fair allocation problem.
To this end, one might wish to implement some fair allocation mechanism from the literature. 
However, the algorithms proposed in the fair division literature are becoming increasingly complex, which often precludes their consideration in university-wide applications. 
Consider for example the CourseMatch algorithm \citep{budish2017coursematch}, used to assign MBA students to classes at the Wharton School of Business. \citeauthor{budish2017coursematch} state that 
``To find allocations, CourseMatch performs a massive parallel heuristic search that solves billions of mixed-integer programs to output an approximate competitive equilibrium in a fake-money economy for courses''. This framework has been applied to the course allocation system at the UPenn Wharton School of Business, which admits approximately $1700$ students to roughly $350$ courses. To understand the system, students are referred to nine instructional videos and a 12-page manual.
While this system may be appropriate for a specialized MBA program, it may not be as effective for university-wide application, especially in settings with non-expert end-users.
Spliddit \citep{goldman2015spliddit} is another application of fair allocation mechanisms to real-world instances; however, it too does not scale well, as its underlying mechanism solves a mixed integer-linear program to find a Nash-welfare maximizing allocation \citep{Caragiannis2016MNW}. 

Indeed, without any restriction on student preferences, computing a fair and efficient allocation is computationally intractable. However, under some reasonable assumptions on student preferences, we can compute fair and efficient allocations in polynomial time. 
If we assume that students simply want to take as many classes as they are allowed to, subject to scheduling constraints and course limits, then student preferences induce \emph{binary submodular} valuations \citep{benabbou2019group, benabbou2021MRF}. 
Submodular valuations exhibit decreasing returns to scale: the larger the bundle agents already have, the less marginal gain they get from additional goods. 
Under binary submodular valuations, each agent values each additional good at either $1$ or $0$. 

Since these valuations correspond to the rank function of some matroid, they are commonly referred to as {\em matroid rank functions} (MRFs) \cite{oxley2011matroids}.
MRFs are highly structured, a structure that has been exploited in the optimization literature \citep{krause2014submodular,oxley2011matroids} and more recently, in fair allocation \citep{Babaioff2021Dichotomous,Barman2021MRFMaxmin,Barman2021MRFTruthful,benabbou2021MRF}.
Most notably, \citet{Babaioff2021Dichotomous} show that when agents have MRF valuations, there is a truthful poly-time algorithm to compute a leximin, envy free up to any good allocation that maximizes both the utilitarian social welfare and Nash social welfare. 
However, analysis of their algorithm places its runtime at roughly $O(n^6 m^{9/2})$ time (where $n$ and $m$ are the number of agents and goods respectively) which significantly hinders scalability. 
Moreover, their algorithm uses complex matroid optimization algorithms as subroutines; this prevents non-expert users from understanding the algorithm. 
Both issues are detrimental to the deployment of such an algorithm for course allocation.
Thus, our main goal is to
\begin{displayquote}
    \textit{develop a simple and fast algorithm to compute fair and efficient allocations under binary submodular valuations.}
\end{displayquote}
\subsection{Our Contribution}\label{subsec:contrib}
Our main contribution is the introduction of a novel algorithm (known colloquially as {\em Yankee Swap}\footnote{Yankee swap is also known as ``Nasty Christmas'' or ``White Elephant''. See \url{https://youtu.be/19ulSNSRKyU} for a discussion.}) for computing fair and efficient allocations for agents with MRF valuations. 

More specifically, Yankee Swap computes prioritized Lorenz dominating allocations. When agents have binary submodular valautions, prioritized Lorenz dominating allocations are known to be leximin, envy free upto any good (EFX) and maximize both the utilitarian social welfare and Nash social welfare. In addition, randomizing over agent priorities can result in {\em ex-ante} fairness guarantees such as ex-ante envy-freeness and ex-ante proportionality. Finally, prioritized Lorenz dominating allocations can be computed in a strategyproof manner. 


Yankee Swap is similar in spirit to the well known round robin algorithm. 
In the round robin algorithm, agents initially start with an empty bundle and proceed in rounds, picking a good in each round one by one from the pool of unallocated goods. 
Agents take sequential actions under Yankee Swap as well. 
Unlike the round robin algorithm, agents have the power to steal goods from other agents if they do not like any unassigned good. 
The agents who have goods stolen from them make up for their loss by either taking an unassigned good or by stealing a good from someone else. 
This procedure induces {\em transfer paths} --- an agent steals a good from someone, who potentially steals a good from someone else and this goes on until someone takes an unassigned good. 
The main difference between our method and the colloquial Yankee Swap is that we only allow an agent to steal a good when a transfer path exists. 
This means that the utility of every agent on the transfer path remains the same (except for the agent that initiates it, whose utility increases by $1$). 

While the algorithm itself is remarkably simple, the analysis is decidedly non-trivial. 
First, we show via a combinatorial argument that transfer paths must exist (\cref{lem:babaioff-paths}). Next, we show that Yankee Swap outputs balanced allocations, in the sense that once an agent cannot initiate transfer paths, the agents that hold goods they want have bundles of approximately the same size (Lemma \ref{lem:yankee-swap-size-upper-bounds}); 
using these facts and careful analysis, we prove \cref{thm:yankee-swap-leximin}. 
We then turn to analyzing the most complex part of Yankee Swap: computing transfer paths. To do so, we construct a \emph{good exchange graph}, and find shortest paths from the goods owned by the least utility agent and an unassigned good (\cref{algo:transfer}); we show that such paths on the good exchange graph correspond to valid path transfers (\cref{thm:path-correctness}). 
We conclude our analysis by comparing the worst-case runtime of Yankee Swap to the current state of the art. 
Yankee Swap runs in $O((n+m)(n + \tau) m^{2})$ time; $\tau$ is the time it takes to compute the valuation $v_i(S)$ for any bundle $S \subseteq G$ and agent $i \in N$. 
This is a significant speedup compared to \citeauthor{Babaioff2021Dichotomous}'s runtime of $O(n^6 m^{7/2} (m + \tau) \log{nm})$ (\cref{subsec:comparison}). 
We further note that, when $m = \Theta(n)$, our algorithm computes a \MAXUSW{} allocation faster than the matroid intersection based method used by \citet{benabbou2021MRF} and \citet{Babaioff2021Dichotomous}.

We are currently in the process of implementing Yankee Swap for assigning classes at the University of Massachusetts, Amherst, working with the university registrar's office, as well as the computer science department at the University of Massachusetts, Amherst.



\subsection{Related Work}\label{subsec:rel-work}
Binary valuation functions (otherwise called {\em dichotomous preferences}) have been studied in various contexts in the economics and computer science literature. More specifically, binary valuations have been studied in mechanism design \citep{Ortega2018dichotomousmechanism,bogomolnaia2005dichotomousmechanism}, auctions \citep{Babaioff2009Auction, Mishra2013dichotomousauction}, and exchange problems \citep{Roth2005Dichotomousexchange, Aziz2019Dichotomousexchange}.

Binary valuations have also been extensively studied in fair allocation. 
\citet{halpern2020binaryadditive} and \citet{Suksumpong2022weightednash} study fair allocation in the restricted setting of binary additive valuations. \citet{Barman2021MRFMaxmin} study the computation of the maximin share under binary submodular valuations. 
\citet{barman2018pathtransfers, Darmann2015MaximizingNP} and \citet{Barman2021approxmnw} study the computation of a max Nash welfare allocation under various binary valuation classes. Our work is not the first to explore transfer paths: \citet{Suksumpong2022weightednash, barman2018pathtransfers} and \citet{Barman2021MRFMaxmin} also utilize transfer path techniques in their algorithm design.

Lastly, \citet{benabbou2021MRF} and \citet{Babaioff2021Dichotomous} study the computation of fair and efficient allocations under matroid rank valuations. \citet{benabbou2021MRF} present the first positive algorithmic result for the valuation class by showing that a utilitarian social welfare maximizing and envy free up to one good allocation can be computed in polynomial time. This result was later significantly improved on by \citet{Babaioff2021Dichotomous} whose work we discuss in detail in \cref{subsec:lorenz-dominating}.

\section{Preliminaries}\label{sec:prelims}
We use $[t]$ to denote the set $\{1, 2, \dots, t\}$. For the sake of readability, for a set $A$ and a good $g$, we replace $A \setminus \{g\}$ (resp. $A \cup \{g\}$) with $A - g$ (resp. $A + g$).

We have a set of $n$ {\em agents} $N = [n]$ and a set of $m$ {\em goods} $G = \{g_1, g_2, \dots, g_m\}$. Each agent $i$ has a {\em valuation function} $v_i:2^G \mapsto \R_+$ --- valuation function $v_i(S)$ corresponds to the value agent $i$ has for the bundle of goods $S$. 
We let $\Delta_i(S,g) \triangleq v_i(S + g) - v_i(S)$ be the marginal utility of agent $i$ from receiving the good $g$, given that they already own the bundle $S$.
Unless otherwise stated, we assume that $v_i$ is a {\em matroid rank function} (MRF). Due to their equivalence, we use binary submodular valuation and matroid rank function interchangeably. More formally, a function $v_i$ is a matroid rank function if
\begin{inparaenum}[(a)]
\item $v_i(\emptyset) = 0$,
\item for every $S\subseteq G$ and every $g \in G$, $\Delta_i(S,g) \in \{0,1\}$, and
\item $v_i$ is submodular: for every $S \subseteq T \subseteq G$ and every $g \in G\setminus T$, $\Delta_i(S,g) \ge \Delta_i(T,g)$.
\end{inparaenum}
Since there may not be a polynomial space representation of these valuation functions, we assume oracle access to each $v_i$: given a bundle of goods $S \subseteq G$, we can compute $v_i(S)$ in at most time $\tau$. 


An {\em allocation} is a partition of the set of goods $X = (X_0, X_1, \dots, X_n)$ where each agent $i$ receives the bundle $X_i$, and $X_0$ consists of the unallocated goods. 
An allocation is {\em non-redundant} (or \emph{clean}) if for every agent $i \in N$, and every good $g \in X_i$, $v_i(X_i) > v_i(X_i - g)$. \citet{benabbou2021MRF} show that for MRF valutions, this is equivalent to having $v_i(X_i) = |X_i|$ for every $i \in N$.
We sometimes refer to $v_i(X_i)$ as the {\em utility} (or {\em value}) of $i$ under the allocation $X$.
For ease of analysis, we treat $0$ as an agent whose valuation function is $v_0(S) = |S|$; this valuation function is trivially an MRF. 
Due to the choice of $v_0$, any clean allocation for the set of agents $N$ is also trivially clean for the set of agents $N + 0$. However, none of the fairness notions we discuss consider the (dummy) agent $0$. 
Several fairness desiderata have been proposed and studied in the literature; three of them stand out:
\begin{description}[leftmargin=0cm]
\item[Envy-Freeness:] An allocation is {\em envy free} if no agent prefers another agent's bundle to their own. 
This is impossible to guarantee when all goods are allocated --- e.g. when there are two agents and only one item. 
Due to this impossibility, several relaxations have been studied in the literature. 
The most popular relaxation of envy freeness is {\em envy freeness up to one good} (EF1) \citep{Budish2011EF1,Lipton2004EF1}. An allocation $X$ is EF1 if no agent envies another agent after dropping some good for the latter agent's bundle, i.e. if for every $i,j\in N$, if $v_i(X_i) < v_i(X_j)$ there exists some $g \in X_j$ such that $v_i(X_i) \ge v_i(X_j-g)$. 
An EF1 allocation can be computed in polynomial time for most realistic valuation classes \citep{Lipton2004EF1}.
More recently, a stronger relaxation called {\em envy free up to any good} (EFX) \citep{Caragiannis2016MNW} has gained popularity: an allocation is EFX if no agent envies another agent after dropping {\em any} good from the latter agent's bundle, i.e. if $v_i(X_i)< v_i(X_j)$ then for \emph{every} $g \in X_j$ $v_i(X_i) \ge v_i(X_j-g)$. 
In contrast to EF1 allocations, the existence of EFX allocations is still an open question for several valuation classes \citep{Plaut2017EFX}.
\item[Maximin Share:] An agent's {\em maximin share} (\MMS) is defined as the value they would obtain had they divided the goods into $n$ bundles themselves and picked the worst of these bundles. 
More formally, 
    \begin{align*}
        \MMS_i =  \max_{X = (X_1, X_2, \dots, X_n)} \min_{j\in [n]} v_i(X_j)
    \end{align*} 
\citet{procaccia2014fairenough} show that agents cannot always be guaranteed their maximin share; however, past works \citep{Kurokawa2018Maxmin} guarantee that every agent receives a fraction of their maximin share. For some $c \in (0, 1]$, an allocation $X$ is $c$-$\MMS$ if for every agent $i \in N$, $v_i(X_i) \ge c\cdot \MMS_i$.
\item[Leximin:] an allocation is leximin if it maximizes the value provided to the agent with least value and conditioned on this, maximizes the value provided to the agent with the second least value and so on. 
While leximin allocations are computationally intractable when agent valuations are unrestricted \citep[Theorem 4.2]{benabbou2021MRF},
they can be computed in polynomial time under MRF valuations \citep{Babaioff2021Dichotomous}.
\end{description}
When envy is the main consideration, an allocation where no agent gets any good is envy free. 
While this is fair, it is very inefficient. 
Therefore, coupled with fairness metrics, algorithms usually guarantee some {\em efficiency} criterion as well. 
We consider two popular notions of efficiency:
\begin{description}[leftmargin=0cm]
\item[Utilitarian Social Welfare:] The {\em utilitarian social welfare} of an allocation is defined as the {\em sum} of the value obtained by each agent i.e. $\USW(X) =\sum_{i \in N} v_i(X_i)$. 

\item[Nash Social Welfare:] The {\em Nash social welfare} of an allocation is defined as the {\em product} of the value obtained by each agent i.e. $\NSW(X) = \prod_{i \in N} v_i(X_i)$. 
\end{description}
Allocations which maximize utilitarian social welfare and Nash social welfare are\textbf{} referred to as \MAXUSW{} and \MNW{} respectively. Since allocations have $\NSW(X) = 0$ when some agent receives no items, we adopt the same convention as \citet{Caragiannis2016MNW}, and first minimize the number of agents with zero utility; subject to that, we maximize the product of positive utilities.

Before we proceed, we preset a simple useful result about matroid rank valuations --- if an agent values the bundle $Y$ more than the bundle $X$, there must be a good $g \in Y$ such that $\Delta_i(X,g) = 1$. Variants of this result have also been shown by \citet{Babaioff2021Dichotomous} and \citet{benabbou2021MRF} and therefore, we omit the proof.

\begin{obs}[\citet{benabbou2021MRF}]\label{obs:exchange}
Suppose that agents have binary submodular valuations. If $X$ and $Y$ are two allocations and $v_i(X_i) < v_i(Y_i)$ for some $i \in N + 0$, there exists a good $g \in Y_i \setminus X_i$ such that $\Delta_i(X_i, g) = 1$. 
\end{obs}



\subsection{Prioritized Lorenz Dominating Allocations}\label{subsec:lorenz-dominating}
We define the {\em sorted utility vector} of an allocation $X$ as $\vec{u}^X = (u^X_1, u^X_2, \dots u^X_n)$ which corresponds to the vector of agent valuations $(v_1(X_1), v_2(X_2), \dots v_n(X_n))$ sorted in ascending order (ties broken arbitrarily). 
An allocation $X$ \emph{Lorenz dominates} the allocation $Y$ (denoted $X \succeq_{\lorenz} Y$) if for all $k \in [n]$, $\sum_{j = 1}^k u^X_j \ge \sum_{j = 1}^k u^Y_j$. 
An allocation $X$ is {\em Lorenz dominating} if for \emph{every} allocation $Y$, we have $X \succeq_{\lorenz}Y$. 

\citet{Babaioff2021Dichotomous} show that when agents have MRF valuations, a non-redundant Lorenz dominating allocation always exists and satisfies several desirable fairness and efficiency guarantees. We formalize this result below.
\begin{theorem}[\citet{Babaioff2021Dichotomous}]\label{thm:lorenz-fair}
When agents have MRF valuations, a non-redundant Lorenz dominating allocation always exists and is \MNW, \MAXUSW, EFX, leximin and $\frac12$-\MMS. 
\end{theorem}

The one minor drawback of the ordering defined above is that it does not distinguish between two allocations with the same sorted utility vector, even when agents receive a different utility in the two allocations. 

\begin{example}\label{ex:lorenz}
Consider a problem instance with two agents $\{1, 2\}$ and three goods $\{g_1, g_2, g_3\}$. The valuation function for each agent $v_i(S) = |S|$ for $i \in \{1,2\}$. Any allocation that gives two goods to one agent and one good to the other is Lorenz dominating, but under one agent 1 may receive two goods, and under another they may receive one.
\end{example}

It is desirable to distinguish between the two allocations described above; if we can create a solution concept where agent $1$ always gets two goods whereas under another agent $2$ gets two goods, randomizing between the two allocations will give us a (random) allocation which is arguably more fair since both agents have the same expected utility.

To this end, \citet{Babaioff2021Dichotomous} introduce a priority order over the set of agents. 
The priority ordering is modelled as a permutation over the set of agents $\pi: N \mapsto [n]$ where agents with a lower value of $\pi$ have a higher priority. 
This ordering is enforced by perturbing the valuations in the original instance.
More formally, to enforce a priority ordering, \citet{Babaioff2021Dichotomous} create a new fair allocation instance, where agent $i$ has the valuation function $v'_i$:
\begin{align*}
    v'_i(S) = v_i(S) +\frac{\pi(i)}{n^2}
\end{align*}
We refer to this problem instance (resp. valuation function $v'$) as the {\em augmented} problem instance (resp. augmented valuation function) with the priority order $\pi$. When $\pi$ is clear from context, we simply refer to this problem instance as the augmented problem instance.
We sometimes refer to the value $\frac{\pi(i)}{n^2}$ as the perturbation.
\citet{Babaioff2021Dichotomous} show that any Lorenz dominating allocation for the augmented instance is also a Lorenz dominating allocation for the original fair allocation instance \citep[Theorem 4]{Babaioff2021Dichotomous} and therefore retain all the desirable fairness properties in \cref{thm:lorenz-fair}.
They refer to this allocation as a Lorenz dominating allocation w.r.t. the priority order $\pi$. 

This priority order can be used to guarantee additional fairness properties --- by randomly choosing this priority order, we can generate allocations that are ex-ante envy-free and ex-ante proportional.
A random allocation $X$ is {\em ex-ante envy-free} if, in expectation, no agent envies another agent i.e. $\E_X[v_i(X_i)] \ge \E_X[v_i(X_j)]$ for all $i, j \in N$. 
Similarly, a random allocation $X$ is {\em ex-ante proportional} if each agent, in expectation, receives a utility greater than the $n$-th fraction of their value for the entire bundle of goods i.e. $\E_X[v_i(X_i)] \ge \frac{v_i(G)}{n}$. 

\citet{Babaioff2021Dichotomous} define the randomized prioritized egalitarian (RPE) mechanism which performs the following steps:
\begin{enumerate}
    \item Choose a priority order $\pi$ uniformly at random.
    \item Elicit the preferences of each agent. If the agent's valuation function is not an MRF, set their value for all bundles to be equal to $0$.
    \item Compute a Lorenz dominating allocation with respect to the ordering $\pi$ for the above elicited preferences.
\end{enumerate}

\citet{Babaioff2021Dichotomous} show that the RPE mechanism computes an ex-ante envy-free and ex-ante proportional allocation. 
In addition, the RPE mechanism is {\em strategyproof}: no agent can get a better outcome by misreporting their valuation function. 
This strategyproofness result is independent of the algorithm used to compute Lorenz dominating allocations; therefore, it applies to Yankee Swap.
More recently, \citet{barman2022groupstrategyproof} show that the RPE mechanism satisfies the stronger notion of {\em group strategyproofness} --- a mechanism is group strategyproof if it is not possible for a set of agents to lie to each get a better outcome. This result applies to Yankee Swap as well.

\section{Yankee Swap}\label{sec:yankee-swap}
We now present a simple, fast algorithm for computing prioritized Lorenz dominating allocations.
The algorithm we propose is known colloquially as a \emph{Yankee swap}. Our objective is to ensure that at every round, the least utility agent receives a \emph{useful} good, i.e. a good for which they have a positive marginal utility.  
We proceed in rounds, and at every round agents have a choice of either taking an unallocated useful good or stealing a useful good from another agent, who then either takes an unassigned useful good or steals one from another agent, and so on. 
Thus, the first agent increases their utility by $1$, and the other agents' utilities remain the same: if a good was stolen from them, they must have recovered their utility by either stealing a good from another agent or by taking an unassigned good.

More formally, we define transfer paths recursively as follows: a \emph{transfer path} in an allocation $X$ is a sequence of agents in $N \cup \{0\}$, $(p_1,p_2,\dots,p_r)$, such that for some good $g \in X_{p_2}$:
\begin{inparaenum}[(a)]
    \item $\Delta_{p_1}(X_{p_1},g) = 1$.
    \item If $X'$ is the allocation that results from moving $g$ from $p_2$ to $p_1$ in $X$; then there exists a transfer path $(p_2, \dots, p_r)$ in $X'$ that does not involve the transfer of the good $g$.
\end{inparaenum}
In other words, there is a set of goods that can be transferred along the path such that agent $p_1$'s utility increases by $1$, agents $p_2,\dots,p_{r-1}$'s utilities are unchanged, and agent $p_r$'s utility decreases by $1$. 
While paths can be cyclic i.e. an agent can be present multiple times in a transfer path sequence, every good may be transferred at most once.

Since goods are transferred at most once, paths can be characterized by a sequence of goods $(g_{i_1}, g_{i_2}, \dots, g_{i_k})$ and an agent $i$ where $g_{i_k}$ gets transferred to the agent that has $g_{i_{k-1}}$, $g_{i_{k-1}}$ gets transferred to the agent that has $g_{i_{k-2}}$ and so on until finally, $g_{i_1}$ gets transferred to agent $i$. Depending on the context, we use both notations of transfer paths in our algorithms and analysis.


\subsection{The Algorithm}\label{sec:yankee-swap-algo}
The Yankee Swap algorithm takes as input a fair allocation instance $(N, G, \{v_i\}_{i \in N})$ and a priority ordering $\pi: N \mapsto [n]$ over the set of agents.
The algorithm first places all goods in $X_0$ --- all goods are initially unassigned --- and has all players playing (denoted by having them in the set $P$). We pick an agent $i \in P$ (who's not agent $0$) with the least utility so far and check if there is a transfer path starting from them and ending at $0$; ties are broken in favor of agents with higher priority. 
If a path exists, we transfer goods backwards along the path giving the agent an additional unit of value. 
Otherwise, we remove them from $P$, at which point their utility can no longer increase. Once all agents are removed from $P$, \cref{algo:yankee-swap} terminates.

\begin{algorithm}
    \caption{Yankee Swap}
    \label{algo:yankee-swap}
    \begin{algorithmic}[1]
        \Require The set of agents $N = [n]$, the set of goods $G$, oracle access to valuation functions $\{v_i\}_{i \in N}$ and a priority order over the agents $\pi: N \mapsto [n]$
        \Ensure A prioritized Lorenz dominating allocation $X$
        \State $X = (X_0, X_1, \dots, X_n) \gets (G, \emptyset, \dots, \emptyset)$
        \State $P \gets N$
        \While{$P \ne \emptyset$}
            \State Let $P' = \argmin\{|X_i|:i \in [n]\}$
            \State Let $i$ be the highest priority agent in $P'$ according to $\pi$
            \State Check if there exists a transfer path in $X$ starting at $i$ which ends at $0$
            \If{a path $(g_{i_1}, g_{i_2}, \dots, g_{i_k})$ exists}
                \State Transfer goods along the path and update $X$        
            \Else
                \State $P\gets P-i$
            \EndIf
        \EndWhile
        \State \Return $X$
    \end{algorithmic}
\end{algorithm}

\subsection{Analysis}\label{subsec:yankee-swap-correctness}
Algorithm \ref{algo:yankee-swap} computes prioritized Lorenz dominating allocations. 
Before we go into the technical details of the proof, we remark that our algorithm is very similar to the round robin algorithm in two ways. 
First, agents in $P$ have roughly the same bundle size at any iteration of the algorithm. 
More formally, all agents in $P$ have bundles that differ in size by at most $1$, with higher priority agents receiving the slightly larger bundles. 
Second, much like how the round robin algorithm is EF1 at every iteration, the Yankee Swap algorithm is Lorenz dominating at every iteration (ignoring the utility of agent $0$ who controls the unassigned items).  

Before we dive into any result about the algorithm, it is important to show sufficient conditions for a transfer path to exist. 
A similar version of the following lemma appears in \citet[Lemma 17]{Babaioff2021Dichotomous} and \citet[Lemma 3.12]{benabbou2021MRF}. 
\begin{lemma}\label{lem:babaioff-paths}
Let $X$ and $Y$ be two non-redundant allocations for the set of agents $N + 0$. Let $S^{-}$ be the set of all agents $i \in N + 0$ where $|X_i| < |Y_i|$, $S^{=}$ be the set of all agents $i \in N + 0$ where $|X_i| = |Y_i|$ and $S^{+}$ be the set of all agents $i \in N + 0$ where $|X_i| > |Y_i|$.
For any agent $i \in S^{-}$, there exists a transfer path from $i$ to some agent $k \in S^{+}$ in $X$.
\end{lemma}
\begin{proof}
Our proof is constructive. We construct the transfer path using the following recursive loop which we denote by $\LOOP(X, Y, i)$:

Since $X$ and $Y$ are non-redundant, using \cref{obs:exchange}, there is some good $g \in Y_i \setminus X_i$ such that $\Delta_i(X_i, g) = 1$. This good must belong to some other agent, say $j \in N+ 0$. If $j \in S^{+}$, we are done. Otherwise, move the good from $j$ to $i$ to create a new allocation $X'$. We now compare $X'$ and $Y$ and define ${S'}^{-}$, ${S'}^{=}$ and ${S'}^{+}$ analogously to ${S}^{-}$, ${S}^{=}$ and ${S}^{+}$. Since $j \in S^{-} \cup S^{=}$ and $j$ lost a good in $X'$, we must have that $j \in {S'}^{-}$. Further, $i \in {S'}^{=} \cup {S'}^{-}$ which means $S'^{+} = S^{+}$. We then repeat this process with the allocations $X'$ and $Y$ and the agent $j$ i.e. we run $\LOOP(X', Y, j)$. 

Note that the above loop must terminate since $\sum_{j \in N + 0} |X'_j \cap Y_j|$ increases by $1$ at every iteration, and is upper bounded by $|G|$. 
Since $S^{+}$ does not change, at some iteration, we take a good from an agent in $S^+$ and the above loop terminates. 

Let $(p_1, p_2, \dots, p_r, k)$ be the sequence of agents we take a good from in the above loop --- we take a good from $p_1$ and give it to $i$, take a good from $p_2$ and give it to $p_1$, and so on.
In order to show that $(i, p_1, \dots, p_r, k)$ forms a transfer path, the last thing we need to show is that $g$ never gets transferred out of $i$ which would imply that goods get transferred at most once. This is easy to see: the only goods that can potentially get transferred in $X'$ are the goods in $Y_j \setminus X'_j$ for any $j \in N + 0$. Since $g \in Y_i \cap X'_i$, it will never get transferred out. Similarly, no good that is transferred to some agent ever gets transferred out.
\end{proof}

When $X$ is the allocation computed by Yankee Swap and $Y$ is the prioritized Lorenz dominating allocation, the above lemma shows that we can, in a way, move closer to $Y$ from $X$ using path transfers. While repeated applications of this lemma to any non-redundant allocation $X$ will ultimately terminate at the prioritized Lorenz dominating allocation, this is not a very efficient process. 

Our next three lemmata establish some important properties of the algorithm.
The first Lemma shows that the allocation maintained by Yankee Swap is always non-redundant; thereby, allowing us to use Lemma \ref{lem:babaioff-paths}.


\begin{lemma}\label{lem:yankee-swap-non-redundant}
At the beginning of every iteration, the allocation $X$ of the Yankee Swap algorithm is non-redundant for the set of agents $N+0$.
\end{lemma}
\begin{proof}
By the definition of transfer paths, agents only ever take unassigned goods/steal goods that they have a positive marginal gain for, i.e. useful goods. 
Thus, after executing a transfer path, the allocation remains non-redundant. 
Since at the beginning of the first iteration the allocation is non-redundant (by our choice of $v_0$), it follows that non-redundancy is maintained throughout.
\end{proof}

The second Lemma shows that picking the least utility agent with highest priority is equivalent to picking the agent with least value according to the augmented valuation $v'$. 
This will help us when analyzing the properties of the allocation output by Yankee Swap with respect to the augmented valuations $v'$.
\begin{lemma}\label{lem:choice-equivalence}
At the beginning of any iteration of \cref{algo:yankee-swap}, let $i$ be the highest priority agent with least utility such that $i\in P$ i.e. let $i$ be the agent chosen by the algorithm to be the starting point of the transfer path. Let $W$ be the allocation at the beginning of the iteration. Then $i$ is the least valued agent in $W$ with respect to the augmented valuation function among all the agents in $P$.
\end{lemma}
\begin{proof}
For any agent $u \in P$, if $|W_u| > |W_i|$, we have $v'_u(W_u) > v'_i(W_i)$ since the augmented valuations add a value smaller than $1$ for both agents compared to the original valuation function and $W$ is non-redundant (\cref{lem:yankee-swap-non-redundant}). 
If $|W_u| = |W_i|$, then we must have $v'_u(W_u) > v'_i(W_i)$ since we chose $i$ as the agent with highest priority with a bundle of size $|W_i|$. 
Since our first constraint on $i$ was that it needed to minimize $|W_i|$, we can never have $|W_u| < |W_i|$.
\end{proof}
Finally, we formalize the round-robin balancedness notion described at the beginning of this subsection.
\begin{lemma}\label{lem:yankee-swap-size-upper-bounds}
Let $i$ be the agent chosen to initiate a transfer path in some iteration of \cref{algo:yankee-swap}. 
Let $W$ be the allocation at the beginning of the iteration. 
If $|W_i| = k$, then the following holds:
\begin{enumerate}
    \item Agents with higher priority than $i$ have a bundle of size at most $k+1$.
    \item Agents with a lower priority than $i$ have a bundle of size at most $k$.
\end{enumerate}
\end{lemma}
\begin{proof}
This result stems from the sequential nature of the allocation. 
Let $t$ be the iteration of the algorithm being examined.

Let $j$ be an agent with higher priority than $i$ with a bundle of size at least $k+2$. 
Consider the start of the iteration $t'$ where $W_j$ moved from a bundle of size $k+1$ to a bundle of size $k+2$. 
Since bundle sizes increase by at most $1$ at every iteration, there is a unique iteration where this event occurred. 
It must also be that $t' < t$.  To be selected on the iteration $t'$, $j$ must have been the agent with least utility $(k+1)$ and highest priority among the agents in $P$. However, we also have that at iteration $t'$, $i$'s bundle had a size of at most $k$ and it was in $P$ as well. 
This is because at iteration $t > t'$, $i$'s bundle had a size of $k$ and $i$ was in $P$; bundle sizes increase monotonically and agents removed from $P$ never get added back. This is a contradiction since it implies $j$ was not the agent with least utility among the agents in $P$ at iteration $t'$; therefore, $j$ cannot have a bundle of size $\ge k+2$.

The case where $j$ has a lower priority than $i$ is handled similarly. 
\end{proof}

We are now ready to prove our main result. 
This is done via a simple contradiction --- if Yankee Swap fails and there exists some agent $i$ with a higher value in the Lorenz dominating allocation, then there must have been a path from $i$ to $0$ that Yankee Swap somehow missed. 
However, showing that this path exists requires showing that Yankee Swap gets sequentially closer to the final solution. 
That is, at every round, the utility of every agent $i$ is no more than their utility under the prioritized Lorenz dominating allocation. 
This proof requires a careful combinatorial argument, and is shown in Lemma \ref{lem:lorenz-dominance}.

\begin{theorem}\label{thm:yankee-swap-leximin}
When agents have MRF valuations, Yankee Swap computes a non-redundant Lorenz dominating allocation with respect to the priority order $\pi$.
\end{theorem}
\begin{proof}

It is easy to see that the algorithm always terminates: at every iteration, an agent is removed from $P$ or $|X_0|$ reduces by $1$. 
Since no item is ever deallocated, and no agents ever return to $P$ once removed, the number of iterations is at most the size of $P$ plus $X_0$ on initialization, i.e. $n+m$. 

We now show that the allocation output by \cref{algo:yankee-swap} is Lorenz dominating with respect to the priority order $\pi$.

Assume for contradiction that the allocation output by \cref{algo:yankee-swap} is not Lorenz dominating with respect to the ordering $\pi$. Let $X$ be the allocation output by Algorithm \ref{algo:yankee-swap} and $Y$ be a Lorenz dominating allocation with respect to the ordering $\pi$ (recall that $Y$ is shown to always exist by \citet{Babaioff2021Dichotomous}). Note that $Y$ is also a Lorenz dominating allocation with respect to the augmented problem instance with priority order $\pi$.

If for all $i \in N$ $v_i(X_i) \ge v_i(Y_i)$ then since $Y$ is \MAXUSW, this means that $v_i(X_i) = v_i(Y_i)$ for all $i$ and $X$ is Lorenz dominating as well with respect to the ordering $\pi$, contrary to our assumption. 
Thus, there must be at least one $i\in N$ such that $v_i(X_i) < v_i(Y_i)$. 
Let $i \in \argmin \{v_i(X_i)\mid v_i(X_i) < v_i(Y_i)\}$. If there are multiple such agents, we pick the one with highest priority. Note that, by an argument similar to \cref{lem:choice-equivalence}, this is equivalent to saying $i$ is the least valued agent in the augmented problem instance such that $v'_i(X_i) < v'_i(Y_i)$.

Consider the iteration of the algorithm where $i$ was removed from $P$. Let $W$ be the non-redundant allocation at the beginning of the iteration (Lemma \ref{lem:yankee-swap-non-redundant}). We have the following Lemma.
\begin{lemma}\label{lem:lorenz-dominance}
For all $h \in N$, we have $|Y_h| \ge |W_h|$.
\end{lemma}
\begin{proof}
Assume for contradiction that this is not true. 
We show that $Y$ does not Lorenz dominate $W$ in the augmented problem instance, which contradicts the fact that $Y$ is Lorenz dominating. 
Let $j$ be the agent with lowest utility under $Y$ such that $|Y_j| < |W_j|$; we break ties by choosing the agent with highest priority. 
Again, since $W$ is non-redundant, this is equivalent to saying that $j$ is the agent with least utility in $Y$ with respect to the augmented problem instance such that $v'_j(Y_j) < v'_j(W_j)$.

Before we delve into the technical details, let us discuss the main idea of the proof. 
We examine the sorted utility vectors of $W$ and $Y$. We show that for any agent $h$ such that $v'_h(Y_h)< v'_j(Y_j)$, then $v'_h(Y_h) = v'_h(W_h)$. 
The `converse' also holds: for any agent $h$, if $v'_h(W_h) < v'_j(Y_j)$ then $v'_h(W_h) = v'_h(Y_h)$. 
We can compare the first index where the sorted utility vectors of $W$ and $Y$ differ; 
we then show that this element is greater in the sorted utility vector of $W$ than in $Y$, i.e. showing that $Y$ does not Lorenz dominate $W$, and yielding a contradiction. 
This is summarized in Figure \ref{fig:sorted-utility-vector}. 
Lastly, in this proof, we will be analyzing the augmented problem instance; unless specifically stated, any claims about preferences can be assumed to be with respect to the augmented valuations.

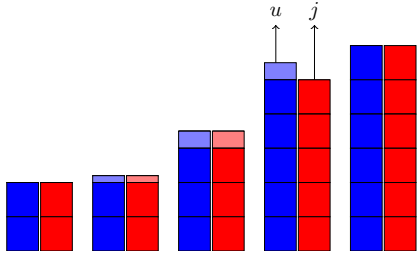
\begin{figure}
    \centering
    \scalebox{0.8}{%
    \begin{tikzpicture}
        \begin{axis} [ybar stacked,
    bar width = 15pt,
    ymin = 0,
    ymax = 10,
    ytick = \empty,
    xtick = \empty,
    bar shift = -8pt,
    hide axis
]
 
\addplot[fill=blue] coordinates {(1,1) (2,1) (3,1) (4,1)(5,1)};
\addplot[fill=blue] coordinates {(1,1) (2,1) (3,1) (4,1)(5,1)};
\addplot[fill=blue] coordinates {(1, 0)(2,0) (3,1) (4,1)(5,1)};
\addplot[fill=blue!50!white] coordinates {(1,0)(2,0.2)(3,0)(4,0)(5,0)};
\addplot[fill=blue] coordinates {(1, 0)(2, 0)(3,0)(4,1)(5,1)};
\addplot[fill=blue!50!white] coordinates {(1, 0)(2, 0)(3,0.5)(4,0)(5,0)};
\addplot[fill=blue] coordinates {(1, 0)(2, 0) (3, 0)(4,1)(5,1)};
\addplot[fill=blue] coordinates {(1, 0)(2,0)(3,0)(4,0)(5,1)};
\addplot[fill=blue!50!white] coordinates {(1, 0)(2,0)(3,0)(4,0.5)(5,0)}; \label{plot_one}
\node  at (axis cs: 3.75, 7){$u$};
\node  at (axis cs: 4.2, 7){$j$};

\draw[->] (axis cs: 3.75, 4.5) edge (axis cs: 3.75, 6.6) (axis cs: 4.2, 4.5) edge (axis cs: 4.2, 6.6);
\end{axis}

\begin{axis} [ybar stacked,
    bar width = 15pt,
    ymin = 0,
    ymax = 10,
    xtick = \empty,
    bar shift = 8pt,
    hide axis
]

\addplot[fill=red] coordinates {(1,1) (2,1) (3,1) (4,1)(5,1)};
\addplot[fill=red] coordinates {(1,1) (2,1) (3,1) (4,1)(5,1)};
\addplot[fill=red] coordinates {(1, 0)(2,0) (3,1) (4,1)(5,1)};
\addplot[fill=red!50!white] coordinates {(1, 0)(2,0.2) (3,0) (4,0)(5,0)};
\addplot[fill=red] coordinates {(1, 0)(2, 0)(3,0) (4,1)(5,1)};
\addplot[fill=red!50!white] coordinates {(1, 0)(2, 0)(3,0.5) (4,0)(5,0)};
\addplot[fill=red] coordinates {(1, 0)(2, 0) (3, 0)(4,1)(5,1)};
\addplot[fill=red] coordinates {(1, 0)(2,0)(3,0)(4,0)(5,1)};
\end{axis}

\end{tikzpicture}
}
    \caption{The idea behind the proof of \cref{lem:lorenz-dominance}: we plot the sorted utility vectors of $W$ (blue) and $Y$ (red). Each large block can be thought of as a utility of $1$ derived from a good and each of the smaller light colored blocks can be thought of as the augmented utilities. 
    We first show that all the values to the left of $j$ are equal in both $W$ and $Y$ and correspond to the same agent. We next show that the bar corresponding to $u$ is taller than that of $j$.}
    \label{fig:sorted-utility-vector}
\end{figure}

We first show that $v'_j(Y_j) < v'_i(W_i)$. 
If $j$ has a higher priority than $i$, $|W_j| \le |W_i|+1$ (\Cref{lem:yankee-swap-size-upper-bounds}). 
By our choice of $j$, $|Y_j| \le |W_j| - 1 \le |W_i|$, which implies that $v'_j(Y_j) < v'_i(W_i)$. 
If $j$ has a lower priority than $i$, $|W_j| \le |W_i|$ (\Cref{lem:yankee-swap-size-upper-bounds}). 
By our choice of $j$, $|Y_j| \le |W_j| - 1 \le |W_i| - 1 < |W_i|$. Since $|W_i|$ is greater than $|Y_j|$ by at least $1$, we have $v'_j(Y_j) < v'_i(W_i)$ irrespective of their priorities.

We observe that for any $h\in N$ whose utility is less than $v'_j(Y_j)$ under $Y$, we have $v'_h(Y_h) = v'_h(W_h)$. 
This is because we pick an agent $j$ with minimum utility under $Y$ for which $v'_j(Y_j) < v'_j(W_j)$. 
Therefore, for any agent $h$ for whom $v'_h(Y_h) <v'_j(Y_j)$ we must have that $v'_h(Y_h) \ge v'_h(W_h)$. 
If the inequality is strict, we have $v'_h(W_h) < v'_h(Y_h) < v'_j(Y_j) < v'_i(W_i)$. 
This implies that $h$ was removed from $P$ before $i$ was removed from $P$ i.e. at the iteration where $i$ was removed from $P$, $h$ was already removed from $P$; otherwise, according to \cref{lem:choice-equivalence}, \cref{algo:yankee-swap} would have chosen $h$ instead of $i$ creating a contradiction.
If $h$ was not in $P$ at the iteration in consideration, $h$ has a utility of $|W_h|$ when \cref{algo:yankee-swap} terminates, i.e. $|W_h| = |X_h|$.
This implies that $v'_h(X_h) = v'_h(W_h) < v'_i(W_i) \le v'_i(X_i)$ and $v'_h(X_h) < v'_h(Y_h)$. 
However, we assumed that agent $i$ is the lowest utility agent among those who have lower utility under $X$ than under $Y$, a contradiction.
We can similarly show that for any agent $h$ with utility less than $v'_j(Y_j)$ in $W$, we have $v'_h(Y_h) = v'_h(W_h)$.

Note that in the valuations $v'$, no two agents have the same value in any allocation due to the perturbation.
Let $j$ be the $\ell$-th least valued agent in $Y$. From our discussion above, the first $\ell-1$ least valued agents in both $W$ and $Y$ have the same utility in both allocations. Let the $\ell$-th least valued agent in $W$ be $u$. From the definition of Lorenz dominance, if $v'_j(Y_j) < v'_{u}(W_{u})$, $Y$ does not Lorenz dominate $W$. 

If $v'_j(Y_j) = v'_{u}(W_{u})$, this implies that $j = u$ since the fractional part of $v'_h(Y_h)$ is unique for every agent $h \in N$. However, if $j = u$, then by assumption we have $v'_j(Y_j) < v'_j(W_j) = v'_{u}(W_{u}) = v'_j(Y_j)$, a contradiction. 
Therefore $v'_j(Y_j) \ne v'_{u}(W_{u})$. 
If $v'_{u}(W_{u}) < v'_j(Y_j)$, then from our discussion, we must have $v'_{u}(W_{u}) = v'_{u}(Y_{u}) < v'_j(Y_j)$. 
This implies that there are $\ell$ agents (the first $\ell$ least valued agents in $W$) which have a lower utility than $j$ in $Y$, contradicting our assumption on $\ell$. 
Therefore, we must have $v'_{u}(W_{u}) > v'_j(Y_j)$ proving that $Y$ does not Lorenz dominate $W$; the sum of the first $\ell$ elements in the sorted utility vector of $W$ is greater than the sum of the first $\ell$ elements in the sorted utility vector of $Y$. 
Since $Y$ is Lorenz dominating in the augmented problem instance, this results in a contradiction and completes the proof.
\end{proof}

We construct a new allocation $Z$ starting at $Y$ and moving goods from agents in $N - i$ arbitrarily to $0$ till $|Z_j| = |W_j|$ for all $j \in N-i$. 
In other words, let $Z_j$ be a size $|W_j|$ subset of $Y_j$ for all $j \in N -i$ and $Z_i = Y_i$. 
Due to \cref{lem:lorenz-dominance}, there will be no $j \in N - i$ where $|Z_j| < |W_j|$. 
Given the dummy agent $0$'s valuation function, this allocation is still non-redundant for the agents $N + 0$. 
Note that $|W_i| < |Z_i|$ and $|W_0| > |Z_0|$. 

Invoking \cref{lem:babaioff-paths} on the allocations $W$ and $Z$, there must be a transfer path from $i$ to $0$ in $W$. 
This stems from the fact that $0$ is the only agent in $S^{+}$. Since $i$ has a valid transfer path, it could not have been removed from $P$ at the current iteration, a contradiction.
\end{proof}
In addition to proving that Yankee Swap outputs a Lorenz dominating allocation, \cref{thm:yankee-swap-leximin} shows that the resulting allocation is non-redundant. 
This is highly desirable in the course allocation setting; indeed, while some works require that algorithms leave no item unallocated, it is \emph{preferable} to have incomplete allocations in course allocation: if the resulting allocation is not non-redundant, then we have assigned students to classes that they either don't want to take, can't fit in their schedule, or they already have been signed up for. In any case, students will have to drop some of their assigned classes, creating additional administrative overheads for both management and the students themselves.

\subsection{Computing Path Transfers}\label{subsec:path-transfers}
The most computationally intensive aspect of \cref{algo:yankee-swap} is finding path transfers. 
We now provide a simple algorithm to compute path transfers. 
A general version of this method can be used to decide if a set is independent in a matroid union \citep{schrijver-book}. 
A similar approach is also used by \citet{Barman2021MRFMaxmin} whose notation we follow.

We define the {\em exchange graph} of an allocation $X$ as a directed graph $\cal G(X) = (G, E)$ on the set of goods $G$. 
If a good $g$ is in  $X_j$ for some agent $j \in N+0$, then an edge exists from $g$ to some other good $g' \notin X_j$ if $v_j(X_j - g + g') = v_j(X_j)$. 
In other words, there is an edge from $g$ to $g'$ if the agent who owns $g$ can replace it with $g'$ with no loss to their utility. In particular, if $g \in X_0$ then there is a directed edge from $g$ to any other good not in $X_0$. This is because agent $0$, who represents the unassigned items, has an additive utility over all items.
We observe that if the good $g$ is stolen from agent $j$, then that good can be replaced with $g'$ iff the edge $(g,g')$ exists.  

In order to check if a transfer path exists from some agent $i$ to $0$, we construct the exchange graph $\cal G(X)$. 
We then compute the set of goods which have a marginal gain of $1$ for the agent $i$ under the allocation $X$ i.e. $F_i(X) = \{g \in G \mid \Delta_i(X_i, g) = 1\}$. 
In the exchange graph, we find the shortest path (if there exists one) from $F_i(X)$ to goods in $X_0$; this can done by adding a source node $s$ in the exchange graph with edges to all the goods in $F_i(X)$ and then using breadth-first search (BFS) to find the shortest path from $s$ to $X_0$. 
From the path in the exchange graph, we can also determine exactly which goods to transfer along the path to update the allocation --- if we have a path $(s, g_1, g_2, \dots, g_k)$, we transfer the good $g_k$ to the agent who has $g_{k-1}$, transfer $g_{k-1}$ to the agent who has $g_{k-2}$ and so on. Finally, we give $g_1$ to $i$; see \cref{algo:transfer}. 

\begin{algorithm}
    \caption{Computing Transfer Paths}
    \label{algo:transfer}
    \begin{algorithmic}
        \Require An allocation $X$ and two agents $i$ and $j$
        \State Construct the exchange graph $\cal G(X)$
        \State Add a source node $s$ to $\cal G(X)$ with an edge to all the goods in $F_i(X)$
        \State Find the shortest path from $s$ to $X_j$ using breadth-first search in $\cal G(X)$
        \If{A path exists}
            \State Return the path $(s, g_{i_1}, g_{i_2}, \dots, g_{i_k})$
        \Else
            \State Return \false
        \EndIf
    \end{algorithmic}
\end{algorithm}

That the shortest path in the exchange graph from $F_i(X)$ to $X_0$ is a transfer path is well known in the matroid literature (albeit using different terminology). This result was first adapted by \citet[Lemma 1]{Barman2021MRFMaxmin} to fair allocation. The missing proofs in this section can be found in the full version of the paper.

\begin{restatable}{lemma}{lempathtransfer}\label{lem:path-transfer-equivalence}
Given a non-redundant allocation $X$, the shortest path from $F_i(X)$ to $X_j$ in the exchange graph $\cal G(X)$ is a transfer path from $i$ to $j$ ($i \ne j$ and $i, j \in N+ 0$) in $X$. 
i.e. transferring goods along the shortest path results in an allocation where $i$'s value for its bundle goes up by $1$, $j$'s value for its bundle goes down by $1$ and all the other agents see no change in the value of their bundles.
\end{restatable}

\citet[Lemma 5]{Barman2021MRFMaxmin} also show that a result similar to \cref{lem:babaioff-paths} is true for paths in the exchange graph as well.
\begin{restatable}{theorem}{thmexchangegraph}\label{thm:exchange-path}
Let $X$ and $Y$ be two non-redundant allocations for the set of agents $N + 0$. Let $S^{-}$ be the set of all agents $i \in N + 0$ where $|X_i| < |Y_i|$, $S^{=}$ be the set of all agents $i \in N + 0$ where $|X_i| = |Y_i|$ and $S^{+}$ be the set of all agents $i \in N + 0$ where $|X_i| > |Y_i|$.
For any agent $i \in S^{-}$, there exists a path in the exchange graph from $F_i(X)$ to $X_k$ for some $k \in S^{+}$.
\end{restatable}


Armed with these two results, we are ready to show correctness i.e. a path exists from agent $F_i(X)$ to $X_j$ in the exchange graph iff a transfer path exists from $i$ to $j$ in the allocation $X$. 

\begin{theorem}\label{thm:path-correctness}
Given a non-redundant allocation $X$, a transfer path exists from agent $i$ to agent $j$ in $X$ if and only if \cref{algo:transfer} outputs a path. Furthermore, the path output by \cref{algo:transfer} is a transfer path.
\end{theorem}
\begin{proof}
The second statement is implied by \cref{lem:path-transfer-equivalence} so we only prove the first statement. 

$(\Rightarrow)$ Assume \cref{algo:transfer} outputs a path. This implies there exists a path from $F_i(X)$ to $X_j$. From \cref{lem:path-transfer-equivalence}, this implies that there is a transfer path from $i$ to $j$

$(\Leftarrow)$ Assume there is a transfer path from $i$ to $j$ in $X$. Let $Y$ be the non-redundant allocation that arises from transferring goods along the transfer path. Apply \cref{thm:exchange-path} to allocations $X$ and $Y$ and the agent $i$. Agent $i$ is clearly in $S^{-}$ and the only agent in $S^{+}$ is $j$ from the definition of the transfer path. Therefore, there must exist a path from $F_i(X)$ to $X_j$ in $\cal G(X)$. This implies that \cref{algo:transfer} outputs a path.
\end{proof}

\subsection{Time Complexity Analysis}
We now analyze the time complexity of \cref{algo:yankee-swap}. 
We represent the allocation $X$ as a binary matrix where $X(i, g) = 1$ if and only if $g \in X_i$. 
Thus, checking if a good belongs to an agent, adding a good to a bundle, and removing a good from a bundle can be done in $O(1)$ time. 

Our time complexity analysis is based on two simple observations. First, the loop in \cref{algo:yankee-swap} runs at most $n+m$ times. 
This is because at each round, we either reduce the size of $X_0$ by $1$ or remove some agent from the game --- the former happens at most $m$ times and the latter happens at most $n$ times; once an agent is removed, they do not return.

Second, \cref{algo:transfer} runs in $O(m^2(n + \tau))$ time (where again, $\tau$ is the maximum time to compute $v_i(S)$). This can be seen by closely examining each step of \cref{algo:transfer}. 
We can construct the exchange graph by examining each possible pair of goods $(g, g')$, finding the agent $i$ whose bundle contains $g$, and checking if $v_i(X_i - g + g') = v_i(X_i)$. This can trivially be done in $O(m^2 (n + \tau))$ time. 
This is the most expensive step of the algorithm. 

Adding a source node and the required edges takes $O(m \tau)$ time; we only need to check if $\Delta_i(X_i, g) = 1$ for each good $g$. Running BFS and finding the shortest path takes $O(m^2)$ time since there are $O(m)$ nodes in the graph. 
Executing path transfers takes $O(nm)$ time --- we iterate through the path and transfer each good to the owner of the previous good; finding the owner takes $O(n)$ time. 

Combining the two observations, we have the following result.
\begin{theorem}\label{thm:runtime}
    \cref{algo:yankee-swap} runs in $O(m^2(n + \tau)(m+n))$ time.
\end{theorem}
A practical application of Yankee Swap in course allocation can likely leverage structural properties of the problem. For example, while real-world instances contain thousands of students and courses, each individual student is typically assigned no more than 5--6 classes. This sparsity most likely allows for more compact representations and further reductions in running time. We leave this analysis to future work
\subsection{Comparison to \citet{Babaioff2021Dichotomous}}\label{subsec:comparison}
The algorithm to compute Lorenz dominating allocations by \citet{Babaioff2021Dichotomous} works as follows: starting with a \MAXUSW{} allocation, repeatedly check if there exists a path transfer that improves the objective function $\sum_{i \in N} \left (v_i(X_i) + \frac{\pi(i)}{n^2}\right )^2$. This is similar to the technique used by \citet{benabbou2021MRF}, who maximize $\sum_{i \in N}v_i(X_i)^2$ via single-item transfers, rather than path transfers.
\citeauthor{Babaioff2021Dichotomous} check for and compute path transfers via matroid intersection algorithms. 
Specifically, they use algorithms for the matroid intersection problem to compute a \MAXUSW{} allocation for the modified problem instance where each agent $i$'s valuation function is upper bounded by some predefined value $k_i$ i.e. $v_i^{\texttt{new}}(S) = \max\{v_i(S), k_i\}$. 
Note that computing the value of a bundle with respect to this modified valuation function still takes $O(\tau)$ time; therefore, from a time complexity perspective, computing a \MAXUSW{} allocation for the valuation profile $v$ is equivalent to computing a \MAXUSW{} allocation for the valuation profile $v^{\texttt{new}}$.

We first analyze the time complexity of computing a \MAXUSW{} allocation and then use it to analyze the time complexity of the algorithm by \citet{Babaioff2021Dichotomous}.  Our analysis uses the matroid intersection algorithm by \citet{Chakrabarty2019MatroidIntersection}, the state-of-the-art algorithm for the matroid intersection problem with access to a rank oracle. The proof can be found in the full version of the paper.

\begin{restatable}{lemma}{lemmatintersection}\label{lem:mat-intersection-time}
When agents have MRF valuations, computing a \MAXUSW{} allocation takes $O(n^2 m^{3/2} (m+\tau) \log{nm})$ time using the matroid intersection problem.
\end{restatable}

Note immediately that when $m = \Theta(n)$, our algorithm computes a \MAXUSW{} allocation faster than the matroid intersection based approach. 
We now present the runtime of \citeauthor{Babaioff2021Dichotomous}'s algorithm.

\begin{theorem}\label{thm:babaioff-runtime}
    The algorithm by \citet{Babaioff2021Dichotomous} computes Lorenz dominating allocations in $\cal O(n^6 m^{7/2}(m + \tau)\log{nm})$ time.
\end{theorem}
\begin{proof}
\citet{Babaioff2021Dichotomous} show that their algorithm computes a \MAXUSW{} solution at most $O(n^4m^2)$ times. Combining this with \cref{lem:mat-intersection-time}, we get the required time complexity.
\end{proof}

Indeed, our algorithm is significantly faster than that of \citet{Babaioff2021Dichotomous}. 
This speedup mainly stems from two sources. 
First, even though \citet{Babaioff2021Dichotomous} use transfer paths, our method of computing them is faster. 
Second, by carefully choosing which transfer paths to check for, we check for much fewer paths. 
Combining these two factors, Yankee Swap offers a significantly better worst-case runtime. 
In particular, when $m = \Theta(n)$, the worst-case runtime of our algorithm is faster by a factor of $O(n^{13/2} \log{n})$.

\section{Conclusions and Future Work}\label{sec:conclusion}
In this work, we show that when agents have binary submodular valuations, Yankee Swap offers a simple and fast method to output fair and efficient allocations. The entire algorithmic framework can be implemented using no more than a few lines of code (Algorithms \ref{algo:yankee-swap} and \ref{algo:transfer}), and offers a far better worst-case runtime guarantee than the current state of the art. The simplicity of Yankee Swap is its key strength: it is easy to understand (even by non-experts) and implement, and can easily be adapted to different settings. This is all achieved while offering the same strong fairness and efficiency guarantees as the current state of the art.

This work highlights the surprising power of combinatorial arguments in computing transfer paths. 
Unlike prior work in this space, we do not invoke complex matroid optimization algorithms, from which the proofs are directly derived, but rather utilize a \emph{simple} approach, `relegating' the complexity to our careful combinatorial analysis. 
We believe that Yankee Swap can be applied to compute justice criteria beyond leximin. 
More specifically, we conjecture that when agents have entitlements (or priority weights) \citep{chakraborty2021weighted}, a modified version of Yankee Swap can be used to compute a weighted leximin allocation. 
We also believe that Yankee Swap can be applied in fair chore allocation problems, and with some adaptations, in settings where agents do not have binary valuations.




\section*{Acknowledgements}
The authors would like to thank anonymous reviewers at WINE 2022 and AAMAS 2023 for useful feedback. The authors would also like to thank Rohit Vaish for feedback on a preliminary version of the paper.

\bibliographystyle{plainnat}
\bibliography{abb,literature}

\begin{thebibliography}{30}
\providecommand{\natexlab}[1]{#1}
\providecommand{\url}[1]{\texttt{#1}}
\expandafter\ifx\csname urlstyle\endcsname\relax
  \providecommand{\doi}[1]{doi: #1}\else
  \providecommand{\doi}{doi: \begingroup \urlstyle{rm}\Url}\fi

\bibitem[Aziz(2019)]{Aziz2019Dichotomousexchange}
Haris Aziz.
\newblock Strategyproof multi-item exchange under single-minded dichotomous
  preferences.
\newblock \emph{Autonomous Agents and Multi-Agent Systems}, 34\penalty0 (1),
  2019.
\newblock ISSN 1387-2532.

\bibitem[Babaioff et~al.(2009)Babaioff, Lavi, and Pavlov]{Babaioff2009Auction}
Moshe Babaioff, Ron Lavi, and Elan Pavlov.
\newblock Single-value combinatorial auctions and algorithmic implementation in
  undominated strategies.
\newblock \emph{Journal of the ACM}, 56\penalty0 (1), 2009.
\newblock ISSN 0004-5411.

\bibitem[Babaioff et~al.(2021)Babaioff, Ezra, and
  Feige]{Babaioff2021Dichotomous}
Moshe Babaioff, Tomer Ezra, and Uriel Feige.
\newblock Fair and truthful mechanisms for dichotomous valuations.
\newblock In \emph{Proceedings of the 35th AAAI Conference on Artificial
  Intelligence (AAAI)}, pages 5119--5126, 2021.

\bibitem[Barman and Verma(2021{\natexlab{a}})]{Barman2021MRFMaxmin}
Siddharth Barman and Paritosh Verma.
\newblock Existence and computation of maximin fair allocations under
  matroid-rank valuations.
\newblock In \emph{Proceedings of the 20th International Conference on
  Autonomous Agents and Multi-Agent Systems (AAMAS)}, pages 169--177,
  2021{\natexlab{a}}.

\bibitem[Barman and Verma(2021{\natexlab{b}})]{Barman2021approxmnw}
Siddharth Barman and Paritosh Verma.
\newblock Approximating nash social welfare under binary xos and binary
  subadditive valuations.
\newblock In \emph{Proceedings of the 17th Conference on Web and Internet
  Economics (WINE)}, page 373–390, 2021{\natexlab{b}}.

\bibitem[Barman and Verma(2022{\natexlab{a}})]{Barman2021MRFTruthful}
Siddharth Barman and Paritosh Verma.
\newblock Truthful and fair mechanisms for matroid-rank valuations.
\newblock In \emph{Proceedings of the 36th AAAI Conference on Artificial
  Intelligence (AAAI)}, pages 4801--4808, 2022{\natexlab{a}}.

\bibitem[Barman and Verma(2022{\natexlab{b}})]{barman2022groupstrategyproof}
Siddharth Barman and Paritosh Verma.
\newblock Truthful and fair mechanisms for matroid-rank valuations.
\newblock In \emph{Proceedings of the 36th AAAI Conference on Artificial
  Intelligence (AAAI)}, pages 4801--4808, 2022{\natexlab{b}}.

\bibitem[Barman et~al.(2018)Barman, Krishnamurthy, and
  Vaish]{barman2018pathtransfers}
Siddharth Barman, Sanath~Kumar Krishnamurthy, and Rohit Vaish.
\newblock Greedy algorithms for maximizing nash social welfare.
\newblock In \emph{Proceedings of the 17th International Conference on
  Autonomous Agents and Multi-Agent Systems (AAMAS)}, page 7–13, 2018.

\bibitem[Benabbou et~al.(2019)Benabbou, Chakraborty, Elkind, and
  Zick]{benabbou2019group}
Nawal Benabbou, Mithun Chakraborty, Edith Elkind, and Yair Zick.
\newblock Fairness towards groups of agents in the allocation of indivisible
  items.
\newblock In \emph{Proceedings of the 28th International Joint Conference on
  Artificial Intelligence (IJCAI)}, pages 95--101, 2019.

\bibitem[Benabbou et~al.(2021)Benabbou, Chakraborty, Igarashi, and
  Zick]{benabbou2021MRF}
Nawal Benabbou, Mithun Chakraborty, Ayumi Igarashi, and Yair Zick.
\newblock Finding fair and efficient allocations for matroid rank valuations.
\newblock \emph{{ACM} Transactions on Economics and Computation}, 9\penalty0
  (4):\penalty0 1--41, 2021.

\bibitem[Bogomolnaia et~al.(2005)Bogomolnaia, Moulin, and
  Stong]{bogomolnaia2005dichotomousmechanism}
Anna Bogomolnaia, Hervé Moulin, and Richard Stong.
\newblock Collective choice under dichotomous preferences.
\newblock \emph{Journal of Economic Theory}, 122\penalty0 (2):\penalty0
  165--184, 2005.
\newblock ISSN 0022-0531.

\bibitem[Budish(2011)]{Budish2011EF1}
Eric Budish.
\newblock The combinatorial assignment problem: Approximate competitive
  equilibrium from equal incomes.
\newblock \emph{Journal of Political Economy}, 119\penalty0 (6):\penalty0 1061
  -- 1103, 2011.

\bibitem[Budish et~al.(2016)Budish, Cachon, Kessler, and
  Othman]{budish2017coursematch}
Eric Budish, G{\'e}rard~P. Cachon, Judd~B. Kessler, and Abraham Othman.
\newblock Course match: A large-scale implementation of approximate competitive
  equilibrium from equal incomes for combinatorial allocation.
\newblock \emph{Operations Research}, 65\penalty0 (2):\penalty0 314--336, 2016.

\bibitem[Caragiannis et~al.(2016)Caragiannis, Kurokawa, Moulin, Procaccia,
  Shah, and Wang]{Caragiannis2016MNW}
Ioannis Caragiannis, David Kurokawa, Herv\'{e} Moulin, Ariel~D. Procaccia,
  Nisarg Shah, and Junxing Wang.
\newblock The unreasonable fairness of maximum nash welfare.
\newblock In \emph{Proceedings of the 17th ACM Conference on Economics and
  Computation (EC)}, page 305–322, 2016.

\bibitem[Chakrabarty et~al.(2019)Chakrabarty, Tat~Lee, Sidford, Singla, and
  Chiu-wai Wong]{Chakrabarty2019MatroidIntersection}
Deeparnab Chakrabarty, Yin Tat~Lee, Aaron Sidford, Sahil Singla, and Sam
  Chiu-wai Wong.
\newblock Faster matroid intersection.
\newblock In \emph{Proceedings of the 60th Symposium on Foundations of Computer
  Science (FOCS)}, pages 1146--1168, 2019.

\bibitem[Chakraborty et~al.(2021)Chakraborty, Igarashi, Suksompong, and
  Zick]{chakraborty2021weighted}
Mithun Chakraborty, Ayumi Igarashi, Warut Suksompong, and Yair Zick.
\newblock Weighted envy-freeness in indivisible item allocation.
\newblock \emph{{ACM} Transactions on Economics and Computation}, 9:\penalty0
  1--39, 2021.

\bibitem[Darmann and Schauer(2015)]{Darmann2015MaximizingNP}
Andreas Darmann and Joachim Schauer.
\newblock Maximizing nash product social welfare in allocating indivisible
  goods.
\newblock \emph{European Journal Operations Research}, 247:\penalty0 548--559,
  2015.

\bibitem[Goldman and Procaccia(2015)]{goldman2015spliddit}
Jonathan Goldman and Ariel~D. Procaccia.
\newblock Spliddit: Unleashing fair division algorithms.
\newblock \emph{{SIGecom} Exchanges}, 13:\penalty0 41--–46, 2015.

\bibitem[Halpern et~al.(2020)Halpern, Procaccia, Psomas, and
  Shah]{halpern2020binaryadditive}
Daniel Halpern, Ariel~D. Procaccia, Alexandros Psomas, and Nisarg Shah.
\newblock Fair division with binary valuations: One rule to rule them all.
\newblock In \emph{Proceedings of the 16th Conference on Web and Internet
  Economics (WINE)}, page 370–383, 2020.

\bibitem[Krause and Golovin(2014)]{krause2014submodular}
Andreas Krause and Daniel Golovin.
\newblock Submodular function maximization.
\newblock In Lucas Bordeaux, Youssef Hamadi, and Pushmeet Kohli, editors,
  \emph{Tractability: Practical Approaches to Hard Problems}, pages 71--104.
  Cambridge University Press, 2014.

\bibitem[Kurokawa et~al.(2018)Kurokawa, Procaccia, and
  Wang]{Kurokawa2018Maxmin}
David Kurokawa, Ariel~D. Procaccia, and Junxing Wang.
\newblock Fair enough: Guaranteeing approximate maximin shares.
\newblock \emph{Journal of the ACM}, 65\penalty0 (2), 2018.

\bibitem[Lipton et~al.(2004)Lipton, Markakis, Mossel, and
  Saberi]{Lipton2004EF1}
R.~J. Lipton, E.~Markakis, E.~Mossel, and A.~Saberi.
\newblock On approximately fair allocations of indivisible goods.
\newblock In \emph{Proceedings of the 5th ACM Conference on Economics and
  Computation (EC)}, page 125–131, 2004.

\bibitem[Mishra and Roy(2013)]{Mishra2013dichotomousauction}
Debasis Mishra and Souvik Roy.
\newblock Implementation in multidimensional dichotomous domains.
\newblock \emph{Theoretical Economics}, 8\penalty0 (2), 2013.

\bibitem[Ortega(2020)]{Ortega2018dichotomousmechanism}
Josué Ortega.
\newblock {Multi-unit assignment under dichotomous preferences}.
\newblock \emph{Mathematical Social Sciences}, 103:\penalty0 15--24, 2020.

\bibitem[Oxley(2011)]{oxley2011matroids}
James Oxley.
\newblock \emph{Matroid Theory}.
\newblock Number~21. Oxford University Press, 2nd edition, 2011.

\bibitem[Plaut and Roughgarden(2017)]{Plaut2017EFX}
Benjamin Plaut and Tim Roughgarden.
\newblock Almost envy-freeness with general valuations.
\newblock \emph{ArXiv}, abs/1707.04769, 2017.

\bibitem[Procaccia and Wang(2014)]{procaccia2014fairenough}
Ariel~D. Procaccia and Junxing Wang.
\newblock Fair enough: Guaranteeing approximate maximin shares.
\newblock In \emph{Proceedings of the 15th ACM Conference on Economics and
  Computation (EC)}, pages 675--692, 2014.

\bibitem[Roth et~al.(2005)Roth, Sönmez, and {Utku
  Ünver}]{Roth2005Dichotomousexchange}
Alvin~E. Roth, Tayfun Sönmez, and M.~{Utku Ünver}.
\newblock Pairwise kidney exchange.
\newblock \emph{Journal of Economic Theory}, 125\penalty0 (2):\penalty0
  151--188, 2005.
\newblock ISSN 0022-0531.

\bibitem[Schrijver(2003)]{schrijver-book}
A.~Schrijver.
\newblock \emph{Combinatorial Optimization - Polyhedra and Efficiency}.
\newblock Springer, 2003.

\bibitem[Suksompong and Teh(2022)]{Suksumpong2022weightednash}
Warut Suksompong and Nicholas Teh.
\newblock On maximum weighted nash welfare for binary valuations.
\newblock \emph{Mathematical Social Sciences}, 117:\penalty0 101--108, 2022.
\newblock ISSN 0165-4896.

\end{thebibliography}

\newpage
\appendix 

\section{Missing Proofs from Section \ref{subsec:path-transfers}}\label{apdx:path-transfers}
The claims with missing proofs in Section \ref{subsec:path-transfers} are very similar to claims from \citet{Barman2021MRFMaxmin}. 

In order to prove them, we must first establish some basic preliminaries about matroid theory; something we did not have to do for proofs before Section \ref{subsec:path-transfers}. 

\subsection{Matroids}\label{apdx:matroids}
A matroid $M$ is defined by a tuple $(E, \cal I)$ where $E$ is referred to as the {\em ground set} and $\cal I$ is a collection of subsets of $E$. Any matroid $(E, \cal I)$ must satisfy the following properties:
\begin{description}
    \item[(I1)] $\emptyset \in \cal I$
    \item[(I2)] If $I \in \cal I$ and $J \subseteq I$, then $J \in \cal I$
    \item[(I3)] If $I_1, I_2 \in \cal I$ such that $|I_1| > |I_2|$, then there exists an element $e \in I_1 \setminus I_2$ such that $I_2 + e \in \cal I$ 
\end{description}
Each $I \in \cal I$ is referred to as an independent set. 

The rank function of a matroid $r_M: 2^E \rightarrow \mathbb{Z}$ is defined as follows for all $S \subseteq E$
\begin{align*}
    r_M(S) = \max_{I \subseteq S: I \in \cal I} |I|
\end{align*}
In other words, the rank of a set is the largest independent set that is contained in it. It is well known that the rank function is a binary submodular function. It is also well known that every binary submodular function corresponds to the rank function of some matroid \citep{schrijver-book, oxley2011matroids}. Another important property of the rank function that will come in handy is that $r_M(S) = |S|$ iff $S$ is independent in $M$.

The union of $k$ matroids $M_1 = (E_1, \cal I_1)$, $M_2 = (E_2, \cal I_2)$, $\dots$, $M_k = (E_k, \cal I_k)$ is defined by a tuple $M = (\bigcup_{i = 1}^k E_i, \cal I)$ where $\cal I$ is given by 
\begin{align*}
    \cal I = \{I_1 \cup I_2 \cup \dots \cup I_k | I_i \in \cal I_i \, \forall i \in [k]\}
\end{align*}
The union of $k$ matroids is known to be a matroid \citep{schrijver-book}. 

Coming back to our fair allocation problem, each agent's valuation function is a matroid rank function and therefore defines a matroid over the set of goods $G$. 
More formally, for each agent $i \in N + 0$, we have a matroid $M_i = (G, \cal I_i)$ whose rank function is given by $v_i$. If an allocation is non-redundant i.e. $v_i(X_i) = |X_i|$ for all $i \in N + 0$, then $X_i \in \cal I_i$ for all $i \in N+0$. 

We are now ready to prove our claims.
\subsection{Proof of Lemma \ref{lem:path-transfer-equivalence}}
In order to prove \Cref{lem:path-transfer-equivalence}, we show the connection between \cref{lem:path-transfer-equivalence} and an equivalent Lemma in \citet{Barman2021MRFMaxmin}. 
The equivalent Lemma (Lemma 1) from \citet{Barman2021MRFMaxmin} is given as follows:
\begin{lemma}\label{lem:barman-lemma-1}
Let $A = (A_0, A_1, \dots , A_n) \in \cal I_0 \times \cal I_1 \times \dots \times \cal I_n$ be an allocation comprised of independent sets and, for agents $i \ne j$, let $P = (g_1, g_2, \dots , g_t)$ be a shortest path in the exchange graph $\cal G(A)$ between $F_i(A_i)$
and $A_j$ (in particular, $g_1 \in F_i(A_i)$ and $g_t \in A_j$). Then, for all $k \in [n] \setminus \{i, j\}$, we have $A_k \Delta P \in I_k$ along with
$(A_i \Delta P) + g_1 \in \cal I_i$ and $A_j - g_t \in \cal I_j$.
\end{lemma}

This is the exact statement from \citet{Barman2021MRFMaxmin}, with the only change being that we include the agent $0$ in the statement. 
Since agent $0$ has a binary additive (and in particular, submodular) valuation, the lemma still holds with it included.

Let us translate each sentence in \Cref{lem:barman-lemma-1} to our notation. 
The first sentence requires each $A_i$ in the allocation $A$ to be an independent set in the matroid $M_i$. In other words, $v_i(A_i) = |A_i|$ for all $i \in N+0$ which is equivalent to requiring $A$ to be non-redundant. 

The second sentence needs no further clarification since we use the same notation as \citet{Barman2021MRFMaxmin}. 
$F_i(A_i)$ refers to the set of goods that gives $i$ a marginal gain of $1$ and $\Delta$ refers to the operation of transferring backwards along the path but does not include the final transfer of $g_1$ to $i$ i.e. $\Delta$ involves giving $g_t$ to the agent that had $g_{t-1}$, $g_{t-1}$ to the agent that had $g_{t-2}$ and so on till we finally {\em discard} $g_1$. 
Therefore, $A_k \Delta P$ is the bundle $A_k$ after goods have been transferred according to the path $P$. Note that $|A_k \Delta P| = |A_k|$ for all $k \ne j$; for $j$, we have $|A_j \Delta P| = |A_j| - 1$.

While overloading the definition of $\Delta$ may lead to ambiguity, it is important to do so to understand the result of \citet{Barman2021MRFMaxmin}. To help with this understanding, for this subsection alone, we use $\Delta$ to refer to the path transfer operation. In all the other parts of this paper, $\Delta$ is used to denote the function that specifies the value of marginal gain. 

The third and final sentence says that after transferring along the path, all the other agents $k \notin \{i, j\}$ still hold independent sets. 
Since their bundles do not change in size, this is equivalent to saying that their valuation does not change. 
It also says that $A_j - g_t$ is an independent set which implies $A_j$ loses a value of only $1$. 
Finally we have $A_i \Delta P + g_1$ is independent, which implies agent $i$ gains a value of $1$. This is because $|A_i \Delta P| = |A_i|$.

Combining these observations, we translate Lemma \ref{lem:barman-lemma-1} as follows:

\lempathtransfer*

\subsection{Proof of Theorem \ref{thm:exchange-path}}
To prove this Theorem, we use a proof very similar to that of Lemma 5 in \citet{Barman2021MRFMaxmin}. They use the proof to prove a similar (albeit not equivalent) claim. We slightly modify the proof to adapt it for our claim. Our proof will also use the following result from \citet[Theorem 42.4]{schrijver-book} and \citet[Lemma 4]{Barman2021MRFMaxmin}.

\begin{lemma}[\citet{schrijver-book, Barman2021MRFMaxmin}]\label{lem:schriver-lemma}
Let $(A_1,\dots, A_n) \in \cal I_1 \times \dots \times \cal I_n$ be a partial allocation comprised of independent sets, i.e.,
$\bigcup_{i = 1}^n A_i$ is independent in the matroid union of these $n$ matroids $M$. Then, $|\bigcup_{i = 1}^n A_i| < r_M(G)$ if and only if there exists a path in the exchange graph $\cal G(A)$ between $\bigcup_{i = 1}^{n} F_i(A_i)$ and some good $h \notin \bigcup_{i = 1}^n A_i$.
\end{lemma}
Putting all these together, we now prove \Cref{thm:exchange-path}.
\thmexchangegraph*

\begin{proof}
For this proof, we use the ideas and notation of \citet[Lemma 5]{Barman2021MRFMaxmin}. The only place where our proof differs is our argument for why we can instantiate Lemma \ref{lem:schriver-lemma}.

We define the set of goods reachable from the set $F_i(X)$ in the exchange graph $\cal G(X)$ as the set $R \subseteq G$. 
We use $B \subseteq N$ to denote the set of agents who have least one good in $R$, i.e., $B = \{k \in [n] : R \cap X_k \ne \emptyset \} + i$; $i$ is explicitly included in $B$. 

For all agents $k \in B$ we have $F_k(X) \subseteq R$. This is trivially true if $k = i$.
Otherwise, since $k \in B$, there exists a good $g \in A_k$ which is reachable from $F_i(X)$. By definition
of the exchange graph, there exists an edge from $g \in X_k$ to all the goods in $F_k(X)$. 
Therefore, $F_k(X)$ is reachable from $F_i(X)$ and $F_k(X) \subseteq R$.

Now, assume for contradiction that a path from $F_i(X)$ to some $g \in X_j$, $j \in S^{+}$ does not exist. This implies that for all $k \in B$ we have $k \in S^{-} \cup S^{=}$. 
Note that $i \in B$ and $i \in S^{-}$ implies $\sum_{k \in B} |X_k| < \sum_{k \in B} |Y_k|$.

Now, consider the restricted allocation $X' = (X_i)_{i \in B}$; in other words, $X'$ is the allocation $X$ restricted to the set of agents $B$. Let $M_B$ be the union of the matroids $\{M_i\}_{i \in B}$. Then, we have $\sum_{k \in B} |X_k| < \sum_{k \in B} |Y_k| \le r_{M_B}(G)$. Since $\sum_{k \in B} |X_k| = |\bigcup_{k \in B} X_k| < r_{M_B}(G)$, we can instantiate Lemma \ref{lem:schriver-lemma} with the allocation $X'$ and the matroids $\{M_i\}_{i \in B}$. We get that there is a path $P$ in the exchange graph $\cal G(X')$ (and, hence, in $\cal G(X)$) from $\bigcup_{k \in B}F_k(X)$ to a good
$h \ne \bigcup_{k \in B}X_k$. Since $R \subseteq \bigcup_{k \in B}X_k$ we get that $h \notin R$.

This creates a contradicts with the definition of the reachable set $R$. if there is a path from $\bigcup_{k \in B}F_k(X) \subseteq R$ to $h$, then $h$ is reachable from $F_i(X)$ as well. Therefore, the result follows.
\end{proof}

\section{Missing Proofs from Section \ref{subsec:comparison}}\label{apdx:comparison}
The following proof requires a basic understanding of matroids. A simple set of definitions have been presented in \cref{apdx:matroids}.
\lemmatintersection*
\begin{proof}
The algorithm to compute a \MAXUSW{} allocation using the matroid intersection problem works as follows: we construct two matroids both defined on the ground set of all good-agent pairs i.e. $E = \{(g, i)\, | \, g \in G, i \in N\}$. Intuitively, any $S \subseteq E$ can be viewed as an allocation where each agent $i$ receives the bundle $X_i = \{g\, |\, (g, i) \in S\}$; note that goods may be allocated to more than one agent in this allocation. 
The first matroid $M_1 = (E, \cal I_1)$ is a simple partition matroid whose independent sets correspond to the set of all allocations which assign each item to at most one agent. 
The second matroid $M_2 = (E, \cal I_2)$ has a rank function $r_2$ defined as $r_2(S) = \sum_{i \in N} v_i(\{g\, |\, (g, i) \in S\})$ i.e. it is the \USW{} of the allocation defined by $S$.

It is easy to see that the largest cardinality set at the intersection of $\cal I_1$ and $\cal I_2$ corresponds to the \MAXUSW{} allocation for the input problem instance. \citet{Chakrabarty2019MatroidIntersection} provide an algorithm to compute the largest cardinality set at the intersection of two matroids that runs in time $\cal O(|E|\sqrt{M} (\log{|E|}) \cal T_{rank})$ where $M$ is the size of the maximum cardinality set in $\cal I_1 \cap \cal I_2$ and $T_{rank}$ is the complexity of computing the rank of both matroids. $M$ is trivially upper bounded by $m$ since each good can get allocated atmost once.

The rank function of the partition matroid can be computed using a simple parse of the set $S$ --- the rank of $S$ in $M_1$ is equal to the number of unique goods present in the allocation $S$.  This can be done in $O(nm)$ time.

Computing the rank function of the second matroid requires the parsing of $S$ to generate the allocation $X$ and then a query to each agent's valuation function to determine the value of each bundle. This can be done in $O(nm + n\tau)$ time. Putting these values in $\cal O(|E|\sqrt{M} (\log{|E|}) \cal T_{rank})$, we get that a \MAXUSW{} allocation can be computed in $O(n^2 m^{3/2} (m+\tau) \log{nm})$ time.
\end{proof}
\end{document}